\documentclass[conference]{IEEEtran}
\IEEEoverridecommandlockouts

\usepackage{cite}
\usepackage{amsmath,amssymb,amsfonts}
\usepackage{algorithmic}
\usepackage{graphicx}
\usepackage{textcomp}
\usepackage{xcolor}
\usepackage{ascmac}
\usepackage{url}
\usepackage{amsthm}
\theoremstyle{definition}

\newtheorem{definition}{Definition}
\newtheorem{theorem}{Theorem}

\def\BibTeX{{\rm B\kern-.05em{\sc i\kern-.025em b}\kern-.08em
    T\kern-.1667em\lower.7ex\hbox{E}\kern-.125emX}}
\begin{document}

\title{Evaluating Company-specific Biases in Financial Sentiment Analysis using Large Language Models}

\author{
\IEEEauthorblockN{Kei Nakagawa}
\IEEEauthorblockA{\textit{Innovation Lab} \\
\textit{Nomura Asset Management Co, Ltd.}\\
Tokyo, Japan \\
kei.nak.0315@gmail.com}
\and
\IEEEauthorblockN{Masanori Hirano}
\IEEEauthorblockA{
\textit{Preferred Networks, Inc.}\\
Tokyo, Japan \\
research@mhirano.jp}
\and
\IEEEauthorblockN{Yugo Fujimoto}
\IEEEauthorblockA{\textit{Innovation Lab} \\
\textit{Nomura Asset Management Co, Ltd.}\\
Tokyo, Japan \\
yu5fujimoto@gmail.com}
}
\maketitle

\begin{abstract}
This study aims to evaluate the sentiment of financial texts using large language models~(LLMs) and to empirically determine whether LLMs exhibit company-specific biases in sentiment analysis. Specifically, we examine the impact of general knowledge about firms on the sentiment measurement of texts by LLMs.
Firstly, we compare the sentiment scores of financial texts by LLMs when the company name is explicitly included in the prompt versus when it is not. We define and quantify company-specific bias as the difference between these scores.
Next, we construct an economic model to theoretically evaluate the impact of sentiment bias on investor behavior. This model helps us understand how biased LLM investments, when widespread, can distort stock prices.
This implies the potential impact on stock prices if investments driven by biased LLMs become dominant in the future.
Finally, we conduct an empirical analysis using Japanese financial text data to examine the relationship between firm-specific sentiment bias, corporate characteristics, and stock performance.
\end{abstract}


\section{Introduction}
\label{sec:introduction}

With the advancement of natural language processing~(NLP) technology, financial and economic text mining—such as analyzing financial statements, analyst reports, social media texts and newspaper articles—has become an essential tool for investment decisions, understanding economic trends, and improving operational efficiency in financial institutions\cite{loughran2020textual,gupta2020comprehensive}. 
In particular, sentiment analysis of financial and economic text data has played a crucial role in understanding investor decision-making and market trends. 
The use of large language models~(LLMs) in sentiment analysis has raised expectations for higher quality sentiment evaluation~\cite{li2023large,deng2023llms}.

However, since financial and investment decisions typically have significant risk, both accuracy and reliability are required in predictions~\cite{nakagawa2018deep,cao2022ai}. 
Therefore, ensuring that LLM outputs are accurate and trustworthy remains a major challenge~\cite{zhao2024revolutionizing,zhou2024large}. 
Moreover, LLMs inherently contain biases related to race, gender, and socioeconomic disparities\cite{gallegos2024bias,rutinowski2024self}, which can prevent fair decision-making and compromise their reliability, potentially leading to significant consequences\cite{jeoung2023stereomap}.

Despite this, there hasn't been enough investigation into whether LLMs have specific biases toward certain companies in their outputs. Since LLMs are trained on large datasets, the information and general opinions about companies within this data could affect the model's results. If these biases exist, the model might consistently produce skewed sentiment scores for certain companies. For example, some companies might get more positive evaluations because of their past performance or media coverage, while others might be more likely to receive negative evaluations~\cite{tetlock2008more,ferguson2015media}.

This issue is critically important for the following reasons: If investors trust the model's output and make decisions based on it, company-specific biases could inappropriately influence investor behavior and market prices. 
Furthermore, in trading algorithms and risk management strategies that utilize sentiment analysis, such biases could distort outcomes\cite{hirshleifer2015behavioral,kumar2015behavioural,manabe2022value}.
Being aware of these challenges can help practitioners make their algorithms and systems more reliable and trustworthy\cite{chuang2022buy}.
Therefore, it is necessary to determine the extent to which LLMs exhibit company-specific biases in sentiment analysis and to quantitatively assess their impact.

Therefore, we aim to empirically determine whether LLMs exhibit company-specific biases in sentiment analysis of financial texts. Specifically, we examine the impact of general knowledge about companies on sentiment measurement by using performance-related texts from financial statements(see Fig. \ref{fig:concept_company_bias}). 
We define and quantify company-specific bias as the difference in sentiment scores when a company’s name is explicitly included in the prompt versus when it is not. 
By systematically changing whether or not company names are included in the prompts, we can observe how these biases appear and measure their strength. This method allows us to empirically show the presence of company-specific sentiment bias in LLMs, giving us insights into how these models might be influenced by existing views of certain companies.

We employ multiple LLMs in the sentiment analysis to determine the extent of bias across different models. By comparing the results from these different models, we can identify patterns of bias that are common across LLMs or unique to specific models, thereby highlighting the importance of model bias in financial analysis. Subsequently, we construct an economic model to theoretically evaluate the impact of company-specific sentiment bias on investor behavior. 
Our model assumes a market where both biased and unbiased investors coexist, aiming to analyze how bias affects stock prices at equilibrium in such a market. 
This model helps us understand how biased LLM investments, when widespread, can distort stock prices.
Finally, we conduct an empirical analysis using Japanese real-world financial text data to identify companies that are more likely to experience company-specific sentiment bias through exposure analysis. We detect any systematic biases in the sentiment scores attributed to specific companies.
Additionally, we examine the impact of these biases on stock performance, exploring whether biased sentiment assessments correlate with actual stock returns.

\section{Related Work}
\subsection{Sentiment Analysis}
Sentiment analysis methods proposed so far in finance  fields can be divided into two main categories: dictionary-based \cite{tetlock2008more, Govindaraj2008} and machine-learning-based \cite{pang-etal-2002-thumbs, hoang-etal-2019-aspect,Sousa2019,zhang-etal-2024-sentiment} methods.

In the former approach, Tetlock \cite{tetlock2008more} examined whether information in news articles could predict future corporate profits and returns by focusing on individual firms. 
The results show that the greater the number of negative words in the news, the lower the future returns of the firm. 
Govindaraj \cite{Govindaraj2008} analyzed tone changes in a specific MD\&A section of Form 10-K, a disclosure document of the company, and demonstrated that these tone changes affect future stock returns. 
They classified some words as positive or negative and then counted the number of positive and negative words in the input articles to estimate their frequency of occurrence. While dictionary-based sentiment evaluation is advantageous in terms of transparency, it has the limitation of struggling with negative or complex sentences.

The later approach is machine-learning-based methods, which were initially proposed using naive Bayes and SVM \cite{pang-etal-2002-thumbs}. 
Among these, the BERT model\cite{devlin2019bert}, which accounts for structural features of sentences including negative sentences, has shown high accuracy in sentiment analysis \cite{hoang-etal-2019-aspect}. 
\cite{Sousa2019} analyzed the sentiment of economic news using fine-tuned BERT models and demonstrated that these models could predict future movements of the Dow Jones Industrial Index. 
Recently, the effectiveness of LLMs has also been reported \cite{zhang-etal-2024-sentiment}. They analyzed the performance of LLMs across a range of sentiment analysis tasks, varying the model used and the form of the prompt. 
Their results suggest that LLMs exhibit competitive or even superior performance compared to fine-tuned smaller language models in sentiment classification tasks. It has also been shown that few-shot learning further improves performance, making it an effective tool in LLM-based sentiment analysis.

Several studies have also been conducted on Japanese economic texts, which is the focus of this study. 
For dictionary-based methods, \cite{sashida2021stock,nakagawa2022investment} focused on the Supervised Sentiment Extraction via Screening and Topic Modeling~(SSESTM\cite{ke2019predicting}) which can make dictionary specialized for stock return forecasting.
They confirmed that the SSESTM model using four years of Japanese articles in the training data gave relatively good results in Japanese stock market.
\cite{nakagawa2022investment} exploited the lead-lag relationship of the sentiments across causally linked companies.
For machine-learning-based methods, \cite{ito2018text}
proposed an interpretable neural network(NN) architecture called gradient interpretable NN~(GINN).
The GINN can visualize both the market sentiment score from a whole financial document and the sentiment gradient scores in concept units. 
They experimentally demonstrated the validity of  text visualization produced by GINN using a real Japanese text dataset.

\subsection{Bias in Language Models}
The increasing use of machine learning models in recent years has highlighted significant challenges related to the biases embedded within these models, both in societal and practical contexts. Extensive research has been conducted on this issue \cite{mehrabi2022surveybiasfairnessmachine}.

Recent studies have explored various types of biases in LLMs. 
Ferara categorized these biases into six types: Demographic, Cultural, Linguistic, Temporal, Confirmation, and Ideological \& Political. They suggested that examining the ethical, social, and practical consequences of each type of bias could contribute to the responsible development and use of language models. Additionally, numerous measures of bias and datasets for assessing bias have been proposed \cite{gallegos2024biasfairnesslargelanguage}.

One approach for estimating bias is the Masked Token method, which involves replacing tokens and comparing the outputs \cite{kurita-etal-2019-measuring, webster2021measuringreducinggenderedcorrelations, jentzsch-turan-2022-gender}. Jentzsch et al. proposed this method for evaluating biases in sentiment analysis tasks. They generated new sentences by replacing gender-specific terms (e.g., "He" with "She") and calculated the difference between the outputs of the original and modified sentences.

Our method employs a similar concept to assess bias but focuses on named entities. Since it is challenging to find direct antonyms for named entities (e.g., replacing "Apple Inc."), our approach involves generating sentences by omitting the corresponding named entities (i.e., company names) and measuring the differences in the outputs.

\subsection{The Knowledge of LLMs about Named Entities}
The distinction between LLMs and smaller language models has been attributed to their level of common knowledge \cite{li-etal-2022-systematic}. This common knowledge, which includes prior knowledge about named entities, is heavily reliant on the training data \cite{chang2023languagemodelbehaviorcomprehensive}.

While this prior knowledge can enhance LLMs performance, it also risks introducing biases. Although biases in general expressions, such as those related to gender and age, have been extensively studied, these studies primarily focus on common expressions like personal pronouns and racial terms. There has been limited research on biases associated with named entities, such as "John" or "Sofia," in Named Entity Recognition \cite{mehrabi2019manpersonwomanlocation}. Thus, there remains a gap in research concerning bias in named entities.

To the best of our knowledge, this study is the first to address bias in named entities, particularly within the context of economic texts.

\section{Company-specific bias}
\label{sec:bias}

In this section, we define the company-specific biases in sentiment analysis.
For each company, we compare the sentiment scores obtained from prompts that either include or exclude the company name. 

\begin{itembox}[t]{\textbf{Prompt \underline{without} Company Name}}
You are a financial analyst.\
Below is a sentence describing financial performance.\
Please rate the sentiment of this sentence on a scale from 1 (bad) to 5 (good).\
Please output only the sentiment score.\
\{\underline{Texts on financial performance}\}
\end{itembox}

\begin{figure}
    \centering
    \includegraphics[width=\linewidth]{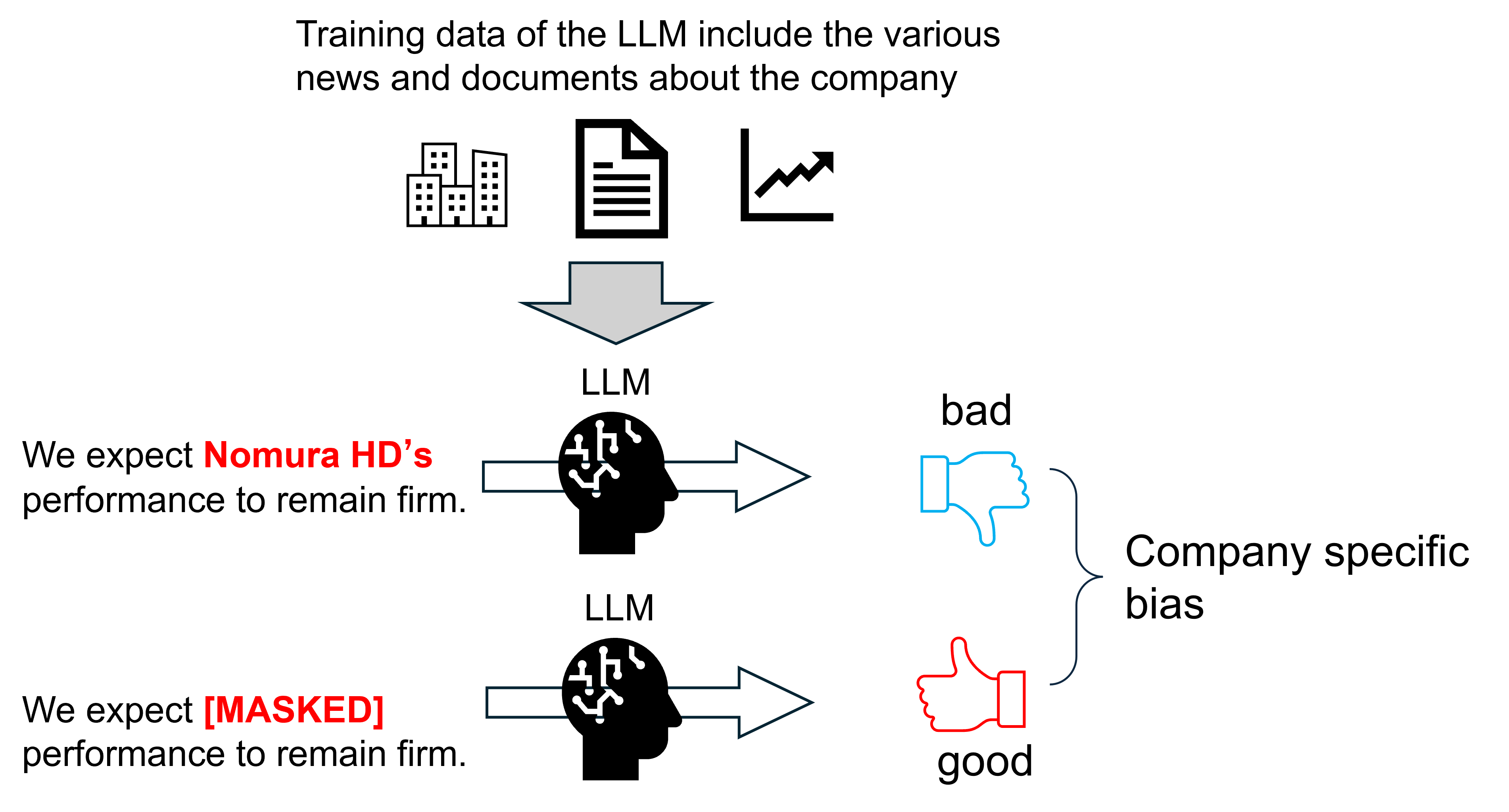}
    \caption{Concept of Company-Specific Bias}
    \label{fig:concept_company_bias}
\end{figure}

\begin{itembox}[t]{\textbf{Prompt \underline{with} Company Name}}
You are a financial analyst.\
Below is a sentence describing the financial performance of \{\underline{Company Name}\}.\
Please rate the sentiment of this sentence on a scale from 1 (bad) to 5 (good).\
Please output only the sentiment score.\
\{\underline{Texts on financial performance}\}
\end{itembox}

Let $s_u$ represent the sentiment score from the prompt without the company name, and $s_b$ represent the score from the prompt with the company name. 

\begin{definition}[company-specific bias]
We define the company-specific bias $\beta$ using the following equation:
\begin{align}\label{eq:bias}
\beta = s_b - s_u
\end{align}
\end{definition}

This bias value $\beta$ provides a quantitative measure of the influence that the inclusion of a company name has on the sentiment analysis. If $\beta$ is positive, it suggests that the LLM exhibits a positive bias toward the company, meaning it assigns a more favorable sentiment when the company name is included. Conversely, if $\beta$ is negative, it indicates that the LLM shows a negative bias, assigning a less favorable sentiment when the company name is included. 

\subsection{Research Questions}

In this study, we aim to address the following research questions on company-specific biases as defined Equation~\ref{eq:bias}:

\textbf{RQ1:~Do we find company-specific bias in sentiment analysis using LLMs?}
We seek to empirically determine whether various LLMs display bias towards specific companies in their sentiment evaluations. 

\textbf{RQ2:~What characteristics are associated with companies that exhibit bias?}
We explore the specific attributes of companies that are prone to bias in sentiment analysis by LLMs. We analyze the characteristics of companies that receive either positive or negative biases and understand how these biases correlate with various company characteristics .

\textbf{RQ3:~How does company-specific bias impact stock performance?}
We investigate the potential impact of company-specific biases on stock performance. By conducting an empirical analysis using real financial data, we examine whether biases in sentiment analysis by LLMs translate into measurable effects on stock prices, particularly in terms of abnormal returns following financial disclosures.

\section{Theoretical Analysis}
\label{sec:theory}

In this section, to get insight of RQ3, we develop an economic model to assess how company-specific sentiment bias influences investor behavior and impacts stock prices in a market with both biased and unbiased investors.

Following the framework of DeLong \textit{et al.} \cite{de1990noise}, we assume that the market consists of two types of investors: those without bias and those with bias. The unbiased investors consider sentiment in their investment decisions, while the biased investors trade based on sentiment bias, using biased sentiment in their decisions. Both types of investors, in the proportions of $1-\mu$ and $\mu$, respectively, make decisions (construct portfolios) at time $t$ and realize payoffs at time $t+1$. The market consists of two types of assets: a risk-free asset that provides a certain return $r$, and a risky asset with an uncertain payoff that provides a real return of $r$. The price of the risk-free asset, considered as the numeraire, is set at 1, while the price of the risky asset is denoted as $p_t$.

If an investor allocates $\lambda_t$ units ($i.e.$, $\lambda_t p_t$) of their initial wealth of 1 to the risky asset, the wealth $W$ at time $t+1$ is given by:
\begin{align}\label{eq:W}
W &= (1 - \lambda_t p_t) (1 + r) + \lambda_t (p_{t+1} + r) \\
&= 1 + r + \lambda_t [p_{t+1} - p_t + (1 - p_t) r]
\end{align}
In the first line, the first term represents the investment in the risk-free asset $(1 - \lambda_t p_t)$ multiplied by its payoff $(1 + r)$, while the second term represents the investment in the risky asset $\lambda_t$ multiplied by its payoff $(p_{t+1} + r)$.
The expected value of $p_{t+1}$ determines the current price $p_t$ of the risky asset. We define the belief about the distribution of $p_{t+1}$ for an unbiased investor as:
\begin{align}\label{pt_ub}
p_{t+1} \sim N(\hat{p}, \sigma^2_p)
\end{align}
For the biased investor, the corresponding belief is defined as:
\begin{align}
&p_{t+1} \sim N(\hat{p} + \beta_t, \sigma^2_p) \label{pt_b} \\
&\beta_t \sim N(\hat{\beta}, \sigma^2_{\beta,t}) \label{rhot_b}\\
&\sigma^2_{\beta,t} = \theta\sigma^2_{\beta,t-1} + \eta_t \label{sigmat_b}
\end{align}
Here, $\eta_t$ is an independent and identically distributed random variable with a mean of zero and variance $\sigma^2_\eta$, and $0 < \theta < 1$ is assumed, meaning that $\sigma^2_{\beta,t}$ is weakly stationary, with its variance remaining constant. The process $\sigma^2_{\beta,t}$ is assumed to follow a first-order autoregressive~(AR) structure, which is a modification of the DeLong \textit{et al.} \cite{de1990noise} framework.

Next, we define the utility function of the investors using a Constant Absolute Risk Aversion~(CARA) form:
\begin{align}\label{eq:utility}
U(W) = -\exp(-2\gamma W)
\end{align}
where $\gamma > 0$ is the coefficient of absolute risk aversion. Under these settings, the equilibrium price is given by the following theorem:

\begin{theorem}[Impact of Bias on Equilibrium Price]\label{thm:eq}
Under the market conditions described above, the equilibrium price $p_t^{*}$ is given by:
\begin{align}
p_t^{*} = 1 + \frac{\mu (\beta_t - \hat{\beta})}{(1 + r)} + \frac{\mu \hat{\beta}}{r} - \frac{\nu\gamma}{r} \frac{\theta \mu^2 \sigma^2_{\beta,t}}{(1 + r)^2} - \frac{\nu\gamma}{r} c
\end{align}
where $\mu, \nu, c$ are constants. The second term indicates how far the current sentiment $\beta_t$ of the biased investors deviates from the long-term average sentiment $\hat{\beta}$. The third term represents the impact of the long-term average sentiment $\hat{\beta}$ on the stock price. The fourth term shows the impact of the variance~(risk) of the sentiment bias of the biased investors on the stock price.
\end{theorem}
\begin{proof}
See the appendix.
\end{proof}

The theorem suggests that theoretically, the impact of bias on stock prices can be either positive or negative. 
In the next section, we will empirically investigate how bias affects stock prices.

In this framework, we can consider the unbiased investors as skilled~(human) analysts, while the biased investors represent LLMs with company-specific sentiment bias. 
Under this interpretation, if biased LLMs become dominant in the market, their influence could significantly impact the equilibrium price. 
Specifically, in a market where biased LLMs play a dominant role, investment decisions would be driven by the sentiment provided by these models. As a result, the sentiment bias of the LLMs would be directly reflected in the market prices. If these LLMs consistently exhibit overly optimistic or pessimistic views, this could lead to upward or downward pressure on stock prices, thereby distorting the price formation in the market. Therefore, this theorem serves as a useful tool for quantitatively evaluating the impact of biased LLMs on the market.

\section{Experiments}
\label{sec:evaluation}
\subsection{Dataset}

In this study, we use the Japanese text data from Summary of Consolidated Financial and Business Results for the Year Ended~(Summary of Financial Results). 
Summary of Financial Results are a standardized financial report that publicly listed companies submit when they announce their earnings. 
The disclosure of financial information in Japan can be based on legal requirements, such as the financial statements published after shareholders' meetings under the Companies Act or annual securities reports filed within three months of the fiscal year-end under the Financial Instruments and Exchange Act.
On the other hand, the Summary of Financial Results is provided under the timely disclosure rules of the stock exchanges, offering investors important financial information as quickly as possible. It is required to be disclosed within 45 days, preferably within 30 days, after the fiscal year-end.

We obtained the data from the Timely Disclosure network~(TDnet\footnote{\url{https://www.release.tdnet.info/inbs/I_main_00.html}}), a service provided by the Tokyo Stock Exchange. 
The data include Summary of Financial Results published between January 2019 and December 2023 by companies listed on the Tokyo Stock Exchange. 
For this study, we focus on the texts on financial performance in the Summary of Financial Results that describe the company's performance, particularly those related to business results (as recorded in the income statement).

An example of such a part of description is provided below, taken from Nissan Motor Co., Ltd.'s financial statement for the fiscal year ending in March 2024:

\begin{itembox}[htpb]{\textbf{Example of texts on financial performance}}
In fiscal year 2023, the global industry volume increased by 8.6\% from the prior fiscal year to 84.54 million units. The Nissan Group (the "Group")'s global retail sales volume increased by 4.1\% from the prior fiscal year to 3,442 thousand units. While retail sales volume in regions excluding China, such as Japan, North America, and Europe, increased by
17.2\% from the prior fiscal year, retail sales volume in the China market declined. The Group’s market share decreased by 0.1 percentage points from the prior fiscal year to 4.1\%.  
\end{itembox}

\subsection{Large Language Model}

This section provides a detailed overview of the LLMs used in our study.

\begin{description}
\item[GPT-4o:]~\\
A multimodal LLM released by OpenAI in May 2024\footnote{\url{https://openai.com/index/hello-gpt-4o/}}. This model is an advanced version of GPT-4\cite{achiam2023gpt}., incorporating both text and image processing capabilities. For this study, however, we focused solely on the text-based functions to align with our analysis needs, excluding the multimodal features.
\item[GPT-3.5-turbo:]~\\  A model released by OpenAI in 2023\footnote{\url{https://openai.com/blog/chatgpt/}}, also known as ChatGPT. This model is known for its efficient processing and conversational capabilities compared to GPT-4.
\item[gemini-1.5-pro-001:]~\\ The latest version of the Gemini series\cite{gemini}, released by Google in 2024. This model is designed for high-performance tasks, offering robust capabilities in text generation and comprehension.
\item[gemini-1.5-flash-001:]~\\  A lightweight version of the Gemini series\cite{gemini}, also released by Google. This version is optimized for faster processing and reduced computational load, making it ideal for applications requiring quick responses with limited resources.
\item[claude-3.5-sonnet:]~\\ The latest model in the Claude series, version 3.5, released by Anthropic as of August 2024\cite{anthropic2024claude}.
\item[claude-3-Haiku:]~\\ A lightweight version of the Claude series released by Anthropic. In version 3, Claude 3 Opus, Sonnet, and Haiku are listed in order of performance, with Haiku being the lightest, fastest, and most cost-effective option\cite{anthropic2024claude}.
\item[Qwen2-7B:]~\\ As a representative of local LLMs, we used Qwen2-7B\cite{qwen2}, a model in the 7B class that can be run efficiently on local machines. This model was selected due to its outstanding performance on the Japanese financial benchmark\cite{Hirano2023-finnlpkdf}, making it a strong candidate for tasks requiring specialized language processing in Japanese.
\end{description}
All models except Qwen2-7B were accessed via API.

\subsection{RQ1:~Evaluation of the company-specific bias}

Here, we focus on the RQ1.
First, we asked each LLMs to rate a sentence related to company performance on a 5-point scale, using prompts that either included or excluded the company name. 
We restricted the output to 10 tokens and checked whether the output included a number between 1 and 5. If multiple numbers appeared, we used the first one as the answer. If no number was present, we marked it as no response and excluded it from the analysis.

We then calculated the bias value defined in Equation \ref{eq:bias} only for cases where we obtained valid ratings from both types of prompts (with and without the company name). The frequency distribution of these bias values is shown in Table \ref{tab:RQ1}. Note that cases where a rating could not be obtained were excluded from the total count, so the total numbers do not match exactly.

\begin{table*}[htbp]
    \centering
    \caption{Results of sentiment bias for each LLMs.}
    \begin{tabular}{lccccccccc}
         Model Name& +4 & +3 & +2 & +1 & $\pm 0$ & -1 & -2 & -3 & -4\\ \hline \hline
         GPT-4o&  &  & 7 & 1328 & 8421 & 486 & 7 &  & \\
         GPT-3.5-turbo &  & 3 & 70 & 1530 & 6904 & 1698 & 32 & 3 & 1\\
         gemini-1.5-pro-001 & & & 10 & 755 & 8573 & 886 & 15 & 1 & \\
         gemini-1.5-flash-001 & & & 1 & 752 & 9067 & 429 & & &\\
         claude-3.5-sonnet & & 2 & 7 & 831 & 8660 & 449 & 1 & & \\
         claude-3-haiku & & 4 & 9 & 296 & 8227 & 1595 & 43 & 18 & \\
         Qwen2-7B & 2 & 4 & 6 & 9 & 1824 & 6 & 7 & 8 & 10 \\ \hline
    \end{tabular}
    \label{tab:RQ1}
\end{table*}

In Table \ref{tab:RQ1}, each bias value, ranging from +4 to -4, indicates the degree of sentiment bias, with positive values representing a more favorable bias and negative values indicating a less favorable bias.
We can observe that each LLM exhibits a bias in approximately 10-20\% of cases, as indicated by the values deviating from $\pm 0$. This confirms the presence of company-specific bias in these models.
Additionally, when comparing the different LLMs, we find that higher-performing models tend to have a higher proportion of $\pm 0$ values, and their bias values are more narrowly distributed. This suggests that better-performing models are more effective at minimizing bias.
It is also noteworthy that the local LLM, Qwen2-7B, has a higher number of instances where it failed to provide a valid evaluation (response) compared to the other models.

\subsection{RQ2:~Evaluation of characteristics of company with bias}
Next, we move onto the RQ2.
In this section, we analyze the relationship between the company-specific bias identified in the previous section and various company characteristics. 
To represent company characteristics, we use 20 factors provided by the MSCI Barra Japan Equity Model~(JPE4)\footnote{\url{https://www.msci.com/www/research-report/model-insight-barra -japan/016268959}}.

The JPE4 model calculates 20 key company characteristics (exposures) for each Japanese company, which are widely used by institutional investors for stock selection and portfolio risk management in the Japanese market. 
In our study, we obtained the exposure values for each stock as of the end of the month before the Results of Financial Summary announcement. 
We then categorized the company-specific bias $\beta$ into three groups: positive, neutral, and negative. For each group, we calculated the average exposure at the time of the earnings announcement.

The average exposure $\overline{E}_k$ for each bias group $k$ (positive, neutral, negative) is calculated using the following formula:

\begin{align}   
    \overline{E}_k = \frac{1}{n_k} \sum_{i=1}^{n_k} E_{i,k}
\end{align}

where $n_k$ is the number of companies in group $k \in \{positive, neutral,negative \}$ , and $E_{i,k}$ represents the exposure of company $i$ in group $k$.

We also calculate the spread of positive and negative

\begin{align}   
    \overline{S} = \overline{E}_{positive} - \overline{E}_{negative} 
\end{align}

\begin{table}[t]
\centering
\caption{Exposures of companies categorized by the bias $\beta$ (positive, neutral, or negative) using GPT-3.5}
\begin{tabular}{lcccc}
&Positive&Neutral&Negative&Spread \\ \hline \hline
Short Term Reversal&-0.01 &-0.05 &-0.11 &0.10 \\
Beta&-0.05 &-0.02 &-0.02 &-0.03 \\
NK225&0.09 &0.10 &0.12 &-0.03 \\
Size&-2.02 &-1.92 &-1.80 & \textbf{-0.22} \\
Residual Volatility&0.36 &0.40 &0.43 &-0.07 \\
Liquidity&-0.26 &-0.25 &-0.16 &-0.10 \\
Momentum&-0.46 &-0.31 &-0.14 & \textbf{-0.32} \\
Non-Linear Size&-1.10 &-0.95 &-0.69 & \textbf{-0.41} \\
Leverage&-0.20 &-0.24 &-0.20 &-0.00 \\
Value&0.59 &0.48 &0.36 & \textbf{0.23} \\
Macro Sensitivity&-0.02 &-0.05 &-0.03 &0.01 \\
Long Term Reversal&0.36 &0.27 &0.23 &0.14 \\
Foreign Sensitivity&-0.53 &-0.53 &-0.59 &0.06 \\
Sentiment&0.07 &0.13 &0.15 &-0.08 \\
Earnings Yield&0.22 &0.17 &0.13 &0.09 \\
Management&-0.17 &-0.11 &-0.07 &-0.11 \\
Industry Momentum&0.01 &0.01 &0.05 &-0.04 \\
Growth&-0.09 &-0.09 &-0.05 &-0.04 \\
Earnings Quality&0.15 &0.07 &-0.00 &0.16 \\
Prospect&-0.61 &-0.61 &-0.59 &-0.02 \\ \hline
\end{tabular}
\label{tab:exp_gpt3.5}
\end{table}

Table \ref{tab:exp_gpt3.5} shows the exposures of companies categorized by the bias $\beta$ (positive, neutral, or negative) using GPT-3.5 for sentiment analysis. 
It also presents the spreads between the companies with positive and negative bias.

\begin{description}
\item[Size, Nonlinear Size:]~\\ We observe that companies with a positive bias have a size exposure of -2.02, while those with a negative bias have a size exposure of -1.80. Similarly, the nonlinear size exposure is -1.10 for companies with a positive bias and -0.69 for those with a negative bias. This suggests that GPT-3.5 tends to negatively evaluate smaller companies when the company name is included in the prompt.
\item[Momentum:]~\\ We find that the momentum exposure is -0.46 for companies with a positive bias, compared to -0.14 for those with a negative bias. This indicates a bias in favor of companies with lower recent stock returns.
\item[Value:]~\\ We observe that the value exposure is 0.59 for companies with a positive bias and 0.36 for those with a negative bias. This suggests a positive bias toward value stocks.
\end{description}

Furthermore, for each LLMs, we calculated the spread between the groups of companies with positive and negative bias, as shown in Table \ref{tab:exposure}.

There are several key observations:

\begin{description}
\item[Variation Across Models:]~\\
We observe that the exposure spreads differ widely across the various LLMs, suggesting that each model may have its own biases towards certain company characteristics. For example, Qwen2-7B shows large positive spreads in factors like Size and Non-Linear Size, which suggests it favors larger companies. In contrast, GPT-3.5 tends to show negative spreads in these factors, indicating a bias against smaller companies.

\item[Consistent Exposure:]~\\
We also notice that some factors, such as Value and Momentum, show consistent biases across multiple models. For instance, GPT-4o and GPT-3.5 display opposite biases in the Value factor, with GPT-4o showing a negative spread and GPT-3.5 showing a positive one. This suggests that different models may interpret similar financial characteristics in different ways.

\item[Model Performance and Exposure:]~\\
We find that models with more advanced architecture, like GPT-4o and Gemini-1.5-pro, generally have smaller spreads, which could mean they provide more balanced evaluations compared to less advanced models. This supports the idea that higher-performing models might be better at reducing bias in sentiment analysis.
\end{description}

Overall, the table shows how different LLMs can exhibit varying levels of bias towards specific company characteristics, which can affect their sentiment analysis results. Understanding these biases is important for correctly interpreting model outputs in financial applications.

\begin{table*}[t]
    \centering
\caption{Results of exposure spread for each LLMs}
  \scalebox{1}{
    \begin{tabular}{lccccccccc} & GPT-4o & GPT-3.5 &  \begin{tabular}{c}
Gemini-\\1.5-pro\end{tabular} &\begin{tabular}{c}
Gemini-\\1.5-flash\end{tabular} &\begin{tabular}{c}
claude-\\3.5-sonnet\end{tabular} &  \begin{tabular}{c}
claude-\\3-haiku\end{tabular} & Qwen2-7b & Average \\ \hline
Short Term Reversal &  -0.006  & 0.104  &  0.040  &-0.101  & 0.022 & 0.056  & 0.329 & 0.070 
\\ 
Beta &  0.018  & -0.033  &0.025  &0.017  &  -0.001 & 0.047  & 0.085 & 0.015 
\\ 
NK225 &  -0.020  &-0.030  &  0.009  &-0.011  &-0.025 & 0.074  &  0.237 & 0.006 
\\ 
Size &  0.120  &-0.219  & 0.068  &0.017  &  0.053 &  0.330 &0.629 & 0.097 
\\ 
Residual Volatility & 0.044  &-0.072  &   -0.092  & -0.072  & -0.124 &0.020  & -0.082 & -0.066 
\\ 
Liquidity &  0.014  &-0.099  & -0.072  &-0.114  &  -0.156 & 0.141  & 0.185 & -0.034 
\\ 
Momentum & 0.332  &-0.319  &   0.014  &-0.049  & -0.088 &  -0.055 &0.124 & -0.010 
\\ 
Non-Linear Size &  0.266  &-0.407  &  0.047  &0.067  & -0.015 &  0.381  &0.787 & 0.129 
\\ 
Leverage & -0.121  &-0.002  &   -0.071  &-0.067  & -0.133 & 0.008  &0.421 & -0.014 
\\ 
Value & -0.328  & 0.230  & 0.038 &-0.016  &   -0.111 & -0.157    &-0.067 & -0.041 
\\ 
Macro Sensitivity & -0.044  &0.014  &   -0.098  &-0.161  & 0.002 &  -0.046  &0.109 & -0.049 
\\ 
Long Term Reversal & -0.163  &0.136  &   -0.005  &0.002  & -0.086 & 0.067  & 0.003 & 0.010 
\\ 
Foreign Sensitivity & -0.022  &0.056  &  0.083  &0.001  & -0.081 & 0.093  &  0.073 & 0.041 
\\ 
Sentiment & 0.171  &-0.082  &0.034 &  0.078 &  -0.142  &  -0.124   & -0.347 &-0.057 
\\ 
Earnings Yield & 0.068  &0.094  &   0.198  &0.052  & 0.097 & -0.091  & 0.255 & 0.082 
\\ 
Management & -0.036  & -0.109  & 0.097  & 0.018  & -0.058 & -0.135  & -0.101 & -0.065 
 \\ 
Industry Momentum & 0.075  &-0.039  &  0.013  &  -0.027  &-0.093 &-0.070  &  0.043 & 0.011\\ 
Growth & 0.259  & -0.039  &  0.080  & 0.058  &0.233 & 0.125 & -0.056 & 0.070 
 
\\ 
Earnings Quality &  -0.115  & 0.157  &  -0.107  &-0.114  & -0.090 &  0.005  &0.297 & 0.014 
 \\ 
Prospect &  -0.022  & -0.018  &  0.024  & 0.001  &0.111 & 0.205  & 0.269 & 0.072 
\\ \hline
    \end{tabular}
    }
    \label{tab:exposure}
\end{table*}

\subsection{RQ3:~Evaluation of Stock Performance}

Finally, we analyze how company-specific bias at the time of the Summary of Financial Results announcements impacts future stock performance. As discussed in Section \ref{sec:theory}, the theoretical impact of such bias on stock performance can be either positive or negative.

Similar to our exposure analysis, we group the announcements by whether the company-specific bias was positive, neutral, or negative, and then calculate the average cumulative abnormal return~(CAR) for each group and each LLM.
This analysis is conducted within the framework of an event study\cite{mackinlay1997event,mackinlay2012econometrics}. 
An event study is a methodology used to assess the effect of a specific event, such as an earnings announcement, on stock prices. By examining abnormal returns~(AR) around the event date, we can determine whether the event has a statistically significant impact on stock performance that deviates from what would be expected under normal market conditions.

To measure the impact of the Summary of Financial Results announcements beyond what can be explained by normal market conditions, we focus on the difference between the actual stock return and the expected return which is known as the AR. 
The AR for stock $i$ on day $t$ after the earnings announcement is given by:

\begin{align}
    AR_{i,t} = r_{i,t} - \hat{r}_{i,t}
\end{align}
where $r_{i,t}$ is the actual return, and $\hat{r}_{i,t}$ is the expected return. 
In this study, we estimate the expected return $\hat{r}_{i,t}$ using the Fama-French 5-factor~(FF5) model~\cite{fama2015five}, which is expressed as:
\begin{align}    
    \hat{r}_{i,t} - r_f = &\alpha_i + \beta_{i,m} (r_m - r_f) + \beta_{i,s} SMB_t + \beta_{i,h}  HML_t \nonumber\\
    & + \beta_{i,r} RMW_t + \beta_{i,c} CMA_t + \epsilon_{i,t}
\end{align}
where $\hat{r}_{i,t}$ is the expected return of stock $i$ at time $t$,$r_f$ is the risk-free rate,$r_m$ is the return of the market portfolio, $SMB_t$ (Small Minus Big) captures the size premium, $HML_t$ (High Minus Low) captures the value premium, $RMW_t$ (Robust Minus Weak) captures the profitability factor,$CMA_t$ (Conservative Minus Aggressive) captures the investment factor\footnote{These factor data were obtained from \url{https://mba.tuck.dartmouth.edu/pages/faculty/ken.french/data_library.html}} and $\epsilon_{i,t}$ is the error term.
$\alpha_i$ is the stock-specific intercept,$\beta_{i,m},\beta_{i,s},\beta_{i,h},\beta_{i,r},$ and $\beta_{i,c}$ are the sensitivities of the stock $i$ to the respective factors.
These parameters are estimated through regression analysis.
We set the estimation window for the model from 130 trading days before the earnings announcement to 11 trading days before (a 120-day period).

We then calculate the AR for each stock over the 60 trading days following the announcement, as well as the CAR from day 0:

\begin{align}
    CAR_{k,t} = \sum_{j=0}^{t} \sum_{i=1}^{n_k} AR_{i,j}
\end{align}
where $n_k$ is the number of companies in group $k \in \{positive, neutral,negative \}$ , and $AR_{i,j}$ represents the AR of company $i$ at time $j$.

\begin{table}[t]
    \centering
    \caption{CAR~(Positive group)}
    \begin{tabular}{lcccc}
Model Name & 1day & 10days & 30days & 60days\\ \hline
GPT-4o & 0.04\% & -0.31\% & -0.46\% & -0.31\%\\
GPT-3.5-turbo & $-0.24\%^*$ & -0.35\% & -0.11\% & 0.25\%\\
gemini-1.5-pro-001 & -0.03\% & $-0.60\%^*$ & $-0.77\%^*$ & -0.61\%\\
gemini-1.5-flash-001 & -0.13\% & -0.06\% & -0.31\% & 0.08\%\\
claude-3.5-sonnet & -0.32\% & $-0.67\%^*$ & $-0.78\%^{**}$ & 0.25\%\\
claude-3-haiku & 0.26\% & 0.72\% & $1.35\%^*$ & $1.71\%^*$\\
Qwen2-7B & 0.79\% & 1.00\% & 0.81\% & -0.94\% \\ \hline
\multicolumn{2}{c}{\footnotesize $*: p<.1, **: p < .05$}
    \end{tabular}
    \label{tab:car_positive}
\end{table}

\begin{table}[]
    \centering
    \caption{CAR~(Negative group)}
    \begin{tabular}{lcccc}
Model Name & 1day & 10days & 30days & 60days\\ \hline
GPT-4o & -0.26\% & $-0.71\%^*$ & $-0.94\%^*$ & -0.37\%\\
GPT-3.5-turbo & 0.02\% & $-0.40\%^*$ & $-0.89\%^{**}$ & $-1.15\%^{**}$\\
gemini-1.5-pro-001 & -0.30\% & $-0.96\%^{**}$ & $-0.96\%^{**}$ & -0.66\%\\
gemini-1.5-flash-001 & 0.20\% & 0.41\% & 0.46\% & 0.60\%\\
claude-3.5-sonnet & 0.07\% & -0.23\% & 0.19\% & 0.72\%\\
claude-3-haiku & 0.06\% & -0.04\% & -0.16\% & -0.18\%\\
Qwen2-7B & -1.45\% & -2.42\% & -2.92\% & -4.62\% \\ \hline
\multicolumn{2}{c}{\footnotesize $*: p<.1, **: p < .05$}
    \end{tabular}
    \label{tab:car_negative}
\end{table}

\begin{table}[t]
    \centering
    \caption{CAR~(Positive group - Negative group)}
    \begin{tabular}{lcccc}
Model Name & 1days & 10days & 30days & 60days\\ \hline
GPT-4o & 0.30\% & 0.40\% & 0.48\% & 0.05\%\\
GPT-3.5-turbo & -0.26\% & 0.05\% & 0.77\% & $1.40\%^{**}$\\
gemini-1.5-pro-001 & 0.27\% & 0.36\% & 0.19\% & 0.04\%\\
gemini-1.5-flash-001 & -0.33\% & -0.47\% & -0.76\% & -0.51\%\\
claude-3.5-sonnet & -0.39\% & -0.43\% & -0.97\% & -0.47\%\\
claude-3-haiku & 0.20\% & 0.77\% & $1.51\%^*$ & $1.89\%^*$\\
Qwen2-7B & 2.25\% & 3.43\% & 3.73\% & 3.68\%\\ \hline
\multicolumn{2}{c}{\footnotesize $*: p<.1, **: p < .05$}
    \end{tabular}
    \label{tab:car_spread}
\end{table}
\begin{figure}
    \centering
\includegraphics[width=\linewidth]{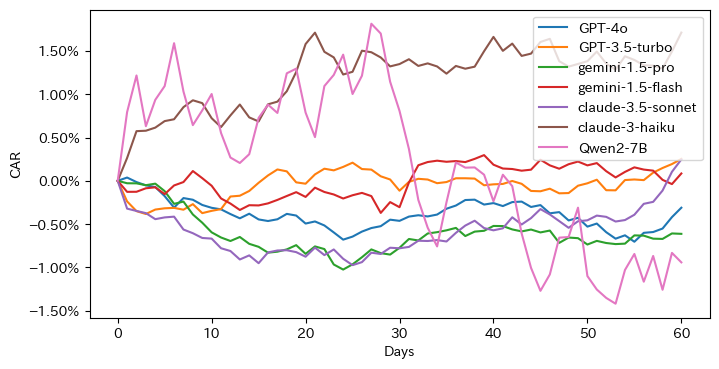}
    \caption{Transition of the CAR~(Positive group) over a 60-day period for various LLMs}
    \label{fig:car_pos}
\end{figure}
\begin{figure}
    \centering
\includegraphics[width=\linewidth]{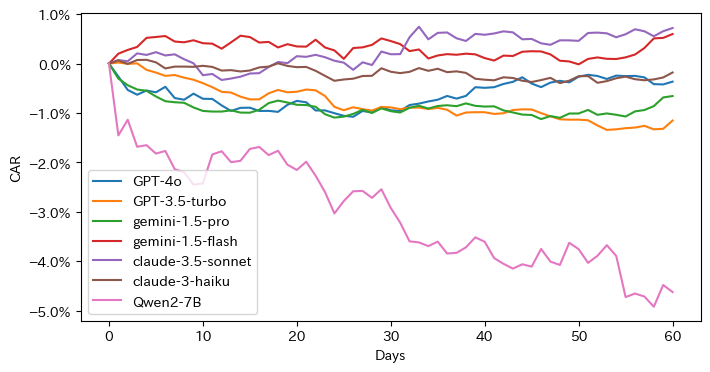}
    \caption{Transition of the CAR~(Negative group) over a 60-day period for various LLMs}
    \label{fig:car_neg}
\end{figure}

Table~\ref{tab:car_positive} and \ref{tab:car_negative} present the CAR for companies categorized by positive and negative company-specific bias, across different LLMs. 
Table~\ref{tab:car_spread} shows the spread of CAR between positive and negative company-specific bias.
The symbols $*: p<.1, **: p < .05$ in Tables indicate the $p$-values from statistical tests to determine whether the values differ significantly from zero.
We analyze these results over various time periods (1 day, 10 days, 30 days, and 60 days) following the announcement of the Summary of Financial Results. 

For Table~\ref{tab:car_positive}, we observed that most models show negative CARs over time. For example, the CAR for GPT-4o and gemini-1.5-pro-001 steadily decreases over the 60-day period, indicating a gradual decline in stock performance. However, Qwen2-7B follows a different pattern. This model initially shows a strong positive CAR (+0.79\%) on the first trading day and maintains a positive CAR (+1.00\%) over the first 10 days. But by the 60-day mark, the CAR turns negative (-0.94\%).

In contrast, from Table~\ref{tab:car_negative}, companies with a negative bias experience a more pronounced decline in CARs across multiple models. GPT-3.5-turbo and gemini-1.5-pro-001, in particular, show significant negative CARs over the 30-day and 60-day periods. The decline in CAR for these models is statistically significant, suggesting that negative bias has a substantial and lasting adverse effect on stock performance. For example, GPT-3.5-turbo records a CAR of -1.15\% at 60 days with \(p < 0.05\). Qwen2-7B again shows a distinct trend, with a steep and continuous decline in CAR from -1.45\% on day 1 to -4.62\% by day 60.
\begin{figure}
    \centering
\includegraphics[width=\linewidth]{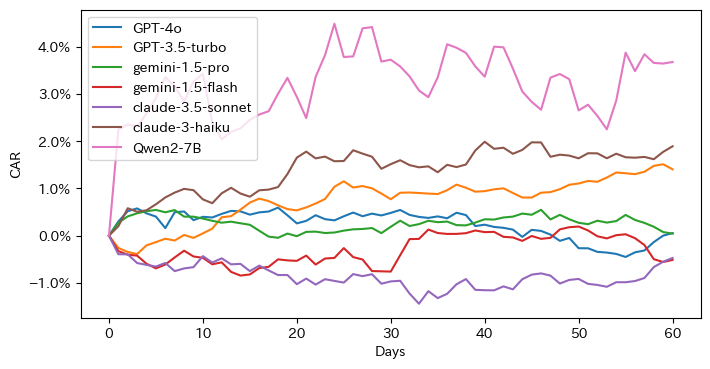}
    \caption{Transition of the CAR~(Positive group - Negative group) over a 60-day period for various LLMs}
    \label{fig:car_spread}
\end{figure}
From Table~\ref{tab:car_spread}, the results show different patterns across the models. For instance, GPT-3.5-turbo shows a significant positive spread at 60 days (+1.40\%, \(p < 0.05\)), suggesting that, despite initial negative market reactions, companies with a positive bias tend to recover better over time compared to those with a negative bias. On the other hand, models like gemini-1.5-flash-001 and claude-3.5-sonnet consistently show negative spreads, indicating that companies with a negative bias continue to underperform relative to those with a positive bias, even over longer periods. In contrast, Qwen2-7B shows a substantial positive spread across all time periods, peaking at +3.68\% at 60 days, suggesting that this model significantly amplifies the differences in stock performance between positively and negatively biased companies.

Fig. \ref{fig:car_pos},\ref{fig:car_neg}, and \ref{fig:car_spread} show the transition of the positive group, negative group and positive - Negative group CARs over a 60-day period for various LLMs, respectively. 
Each line represents the CAR performance of a specific LLM. 
The plot highlights the differing capabilities of these LLMs in capturing sentiment-driven market reactions over time.

These results suggest that bias can have both positive and negative effects on stock prices as Theorem~\ref{thm:eq} indicates.

\section{Conclusion}
\label{sec:conclusion}
In this study we investigate to evaluate the sentiment of financial texts using LLMs and to empirically determine whether LLMs exhibit company-specific biases in sentiment analysis.
The contributions of this study are as follows:
\begin{description}
\item[\textbf{Empirical demonstration of company-specific bias:}]~\\
We empirically demonstrated the presence of company-specific bias in sentiment analysis using multiple LLMs. Specifically, we compared sentiment scores generated with prompts that include company names to those without, quantifying the impact of company names on sentiment evaluation.
\item[\textbf{Development of an economic model for biased investors}]~\\
We developed a theoretical economic model to evaluate the impact of company-specific bias on investor behavior and market prices. This model provides a theoretical foundation for analyzing how biased and unbiased investors coexist in the market and how this affects stock prices.
\item[\textbf{Assessment of bias impact:}]~\\ We conducted an empirical analysis using real financial data to examine the extent to which company-specific bias affects company characteristics and stock performance.
\end{description}

For future research, while we focused on data from the Japanese stock market in this study, it is also important to investigate company-specific biases in different markets, such as those in North America, Europe, or Asia. Since each market may have distinct characteristics and investor behaviors, we believe it is necessary to quantify and evaluate these differences.
Additionally, although LLMs generally support multiple languages, we find it crucial to examine the presence of biases across different languages. For example, we suggest assessing the consistency of sentiment evaluations between English and Japanese and determining whether unique biases exist in each language.


In this Appendix we give the proof of the theorem~\ref{thm:eq}.
\begin{proof}
When prices follow a normal distribution, the expected utility maximization problem in Equation \eqref{eq:utility} is equivalent to maximizing the following expression\cite{markowitz1991foundations}:
\begin{align}
    \omega - \gamma \sigma^2_\omega
\end{align}
where $\omega$ represents the expected wealth, and $\sigma^2_\omega$ is the variance of wealth one period ahead.

The unbiased investor chooses the portfolio of the risky asset, $\lambda^u_t$, to maximize their expected utility:
\begin{align}
    E(U) &= c_0 + \lambda^u_t [r + E_t(p_{t+1}) - p_t(1 + r)] \nonumber\\
    &- \gamma (\lambda^u_t)^2 E_t(\sigma^2_{p_{t+1}})
\end{align}
where the expected volatility of the risky asset’s price from time $t$ to $t+1$ is defined as:
\begin{align}
    E_t(\sigma^2_{p_{t+1}}) = E_t[(p_{t+1} - E_t(p_{t+1}))^2]
\end{align}
On the other hand, the biased investor selects the portfolio of the risky asset, $\lambda^b_t$, to maximize their expected utility:

\begin{align}
E(U) &= c_0 + \lambda^b_t [r + E_t(p_{t+1}) - p_t(1 + r)] \nonumber \\
     & - \gamma (\lambda^b_t)^2 E_t(\sigma^2_{p_{t+1}}) + \lambda^b_t \beta_t
\end{align}

Solving the above optimization problem using the first-order condition, the portfolio holdings of the risky asset for the unbiased and biased investors are:
\begin{align}
&\lambda^u_t = \frac{r + E_t(p_{t+1}) - p_t(1 + r)}{2\gamma E_t(\sigma^2_{p_{t+1}})} \label{lambda_u}\\
&\lambda^b_t = \frac{r + E_t(p_{t+1}) - p_t(1 + r)}{2\gamma E_t(\sigma^2_{p_{t+1}})} + \frac{\beta_t}{2\gamma E_t(\sigma^2_{p_{t+1}})} \label{lambda_b}
\end{align}

The equilibrium price $p_t$ is derived from the market clearing condition, where the sum of the risky asset holdings of both the unbiased and biased investors equals 1 ($\lambda^u_t + \lambda^b_t = 1$). Using Equations \eqref{lambda_u} and \eqref{lambda_b}, we obtain:
\begin{align}
    p_t = \frac{1}{1 + r} [r + E_t(p_{t+1}) - 2\gamma E_t(\sigma^2_{p_{t+1}}) + \mu \beta_t]
\end{align}
To compute the equilibrium price, $p_{t+1}$ is recursively calculated. Based on the biased investor’s misperception described in Equations \eqref{pt_b} and \eqref{sigmat_b}, we assume:
\begin{align}
\sum_{k=1}^\infty \frac{1}{(1 + r)^k} E_t(\sigma^2_{p_{t+k}}) = \frac{\nu}{r} E_t(\sigma^2_{p_{t+1}})
\end{align}
which simplifies to:
\begin{align}
p_t = 1 + \frac{\mu (\beta_t - \hat{\beta})}{1 + r} + \frac{\mu \hat{\beta}}{r} - \frac{2\nu\gamma}{r} E_t(\sigma^2_{p_{t+1}})
\end{align}
Here, since $\hat{\beta}, \gamma, r, \mu, \nu$ are constants, only the second and third terms on the right-hand side are random variables. Therefore, the variance is given by:
\begin{align}
\sigma^2_{p_t} = \frac{\mu^2 \sigma^2_{\beta,t}}{(1 + r)^2} + \frac{4\nu^2 \gamma^2}{r^2} \text{var}[E_t(\sigma^2_{p_{t+1}})]
\end{align}
By iterating $E_t(\sigma^2_{p_{t+1}})$ and eliminating non-stationary prices ($\sigma^2_{p,\infty}, \infty$), we obtain:
\begin{align}
    E_t(\sigma^2_{p_{t+1}}) = \frac{\theta \mu^2 \sigma^2_{\beta,t}}{(1 + r)^2} + c
\end{align}
where \(c\) is a positive constant.

Thus, the assumption:
\begin{align}
\sum_{k=1}^\infty \frac{1}{(1 + r)^k} E_t(\sigma^2_{p_{t+k}}) = \frac{\nu}{r} E_t(\sigma^2_{p_{t+1}})
\end{align}
can be validated by determining the constant $\nu$.

Finally, the equilibrium price $p_t^*$ is given by:
\begin{align}
p_t^* = 1 + \frac{\mu (\beta_t - \hat{\beta})}{(1 + r)} + \frac{\mu \hat{\beta}}{r} - \frac{2\nu\gamma}{r} \frac{\theta \mu^2 \sigma^2_{\beta,t}}{(1 + r)^2} - \frac{2\nu\gamma}{r} c
\end{align}
\end{proof}

\end{document}